\documentclass{llncs}[11pt]
\usepackage{graphicx}
\usepackage{amssymb,amsmath}
\usepackage{psfrag}
\usepackage{fullpage}
\usepackage[noend]{algorithmic}
\usepackage{algorithm,verbatim}

\newcommand{\REM}[1]{}

\newcommand {\mStar} {{m}^{\ast}}
\newcommand {\mStarV} {{m}^{\ast}_{v}}
\newcommand {\mStarX} {{m}^{\ast}_{x}}
\newcommand {\mVmax} {{\bar{m}}^{\ast}}

\newcommand {\weight} {{\bf {w}}}
\newcommand {\bc} {\mbox{BC}}
\newcommand {\DAG} {\mbox{DAG}}
\newcommand {\newpaths} {\widehat{\sigma}'_{sv}}

\newcommand {\A} {\mathcal{A}}

\newcommand{\Grev}{G_R}

\newcommand{\Ein}{E_i}

\newcommand{\Eout}{E_o}

\usepackage{url}
\urldef{\mailid}\path|{meghana,cavia,vlr}@cs.utexas.edu|

\title{Betweenness Centrality -- Incremental and Faster
\thanks{ This work was supported in part by NSF grant CCF-0830737}
}

\author {Meghana Nasre~\thanks{Indian Institute of Technology Madras, India. Email: {\tt meghana@cse.iitm.ac.in}}
\and Matteo Pontecorvi~\thanks{University of Texas at Austin, USA. Email: {\tt cavia@cs.utexas.edu}}
\and Vijaya Ramachandran~\thanks{University of Texas at Austin, USA. Email: {\tt vlr@cs.utexas.edu}}
}
\institute{}
\begin{document}
\maketitle
\begin{abstract}
We consider the incremental computation of the betweenness centrality (BC) of
all vertices in a graph $G = (V, E)$, 
directed or undirected, with positive real edge-weights.
The current widely used algorithm 
is the Brandes algorithm that runs in $O(mn + n^2 \log n)$ 
time, where $n = |V|$ and $m = |E|$. 

 We present an incremental algorithm that updates the BC
score of all vertices in $G$ when a new edge is added to $G$, or the weight of 
an existing edge is reduced. Our incremental algorithm runs in $O(m' n + n^2)$ 
time, 
where $m'$ is bounded by $m^*=|E^*|$, and $E^*$ is
the set of edges that lie on a shortest path in $G$. 
We achieve the same bound for the more general incremental update of a 
vertex $v$, where the edge update can be performed on any subset of edges
incident to $v$.

Our incremental algorithm is the first algorithm that is asymptotically 
faster on sparse graphs than recomputing with the Brandes algorithm even for a single edge update.
It is also likely to be much faster than the Brandes algorithm 
on dense graphs since $m^*$ is  often close to linear 
in $n$.

 Our incremental algorithm is very simple, and we
give an efficient cache-oblivious implementation that incurs
$O(scan(n^2) + n \cdot sort (m'))$ cache misses, where $scan$ and
$sort$ are well-known measures for efficient caching. 
We also give a static BC algorithm that runs in time 
$O(m^* n + n^2 \log n)$, which is faster
than the Brandes algorithm on any graph with $m= \omega(n \log n)$
and $m^* = o(m)$.

\end{abstract} 

\section{Introduction}\label{sec:intro}

Betweenness centrality (BC) is a widely-used measure 
in the analysis of large complex networks.  The BC of a node $v$ 
in a network is the fraction of all shortest paths in the network that go through $v$,
and this measure is often used as an index that  determines the relative importance of
$v$ in the network.
Some applications of 
BC
include analyzing social interaction networks \cite{KA12},
identifying lethality in biological networks \cite{PCW05},
and identifying key actors in terrorist networks \cite{Krebs02,Coffman}. 

Given the changing nature of the networks under consideration,
it is desirable to have algorithms that compute 
BC
faster than computing it from scratch after every change. Our main contribution
is the first incremental algorithm for computing 
BC
after an incremental update on an edge or on a vertex
that is provably faster on sparse graphs than the widely used 
static algorithm by Brandes~\cite{Brandes01}.  
By an {\it incremental update} on an edge $(u,v)$ we mean a decrease in
the weight of an existing edge $(u,v)$, or the addition of a new edge 
$(u,v)$ with finite weight if $(u,v)$ is not present in the graph;
in an incremental vertex update, updates can occur on any subset of edges 
incident to $v$, including the addition of new edges.

Let $G = (V, E)$ be a 
graph with positive real edge weights. Let
$n=|V|$ and $m=|E|$.  
To state our result we need the following definitions. 
For a vertex $x \in V$,
let $\mStarX$ denote the number
of edges that lie on shortest paths through $x$.
Let $\mVmax$ denote the average over all $\mStarX$, i.e.,
 $\mVmax = \frac{1}{n} \sum_{x \in V}{\mStarX}$. 
Finally, let $\mStar$ denote the total number of edges that lie on shortest paths
in $G$. For our incremental bound, we consider the maximum of each of these terms in
the two graphs before and after the update.
Here is our main result.

\begin{theorem}
\label{thm:main1-intro}
After an incremental update on a vertex $v$ 
in a directed or undirected
graph with positive edge weights, the
betweenness centrality of all vertices 
can be recomputed in $O(m' \cdot  n + n^2)$ time, where
$m'= \mVmax + m_v^{\ast}$.
\end{theorem}

Our method is to efficiently maintain the single source shortest paths
(SSSP) directed acyclic graph (DAG) rooted at each source $s\in V$,  and
therefore ensure that after every change we only examine the edges 
that lie on shortest path DAGs. 
Since $m_v^{\ast} \leq m^*$ for all $v\in V$, and
$\mStar \le m$,  the worst case time for our incremental
algorithm is bounded by $O(mn + n^2)$, which
is a $\log{n}$ factor improvement over Brandes' algorithm on
sparse graphs. Moreover, our bound in terms of $\mVmax$ and $m_v^*$ results
in much better bounds for many dense graphs.
For example, as noted in \cite{KKP93}, it
is known \cite{FG85,HZ85,LRP89} that $\mStar = O(n \log{n})$
with high probability in a complete graph where edge weights are chosen from a
large class of probability distributions, including the uniform distribution on the range $\{1, \ldots, n^2 \}$, and
our algorithm is much faster than~\cite{Brandes01} 
on these graphs.

Our algorithm is very simple, especially for the edge update case,
and we present an efficient
cache-oblivious version that incurs $O(scan(n^2) + n \cdot sort(m'))$ 
cache misses.
Here, for a cache of size $M$ that can hold $B$ blocks, $scan(r)= r/B$ and
$sort(r)= (r/B) \cdot  \log_M r$ with a tall cache ($M \geq B^2$). Both $scan$ and
$sort$ are measures of good caching performance (even though $sort(r)$ 
performs $r\log r$ operations, the base of $M$ in the $\log$ makes $sort(r)$
preferable to, say $r$ cache misses). In contrast, the Brandes algorithm
calls Dijkstra's algorithm, which is affected by unstructured accesses
to adjacency lists that lead to large caching costs (see, e.g., 
\cite{MehlhornM02}).


In our incremental algorithm, we  assume that the BC score 
for all vertices has been computed and the SSSP DAG rooted at every vertex
is available. This can be achieved
by running Brandes' algorithm.  
We present an alternate algorithm based
on the Hidden Paths algorithm in Karger~et~al.~\cite{KKP93}, which computes
APSP in
$O(m^* n + n^2 \log {n})$ time.
This leads to
a static BC algorithm with the same time bound as the Hidden Paths algorithm,
which is faster than Brandes' when $\mStar = o(m)$
and $m= \omega(n \log n)$. We also note that substituting Pettie's
$O(mn + n^2 \log\log n)$~\cite{Pettie04} all-pairs shortest paths (APSP) algorithm for directed graphs
or the $O(mn \cdot \log \alpha(m,n))$~\cite{PR05}
APSP algorithm for undirected graphs (where $\alpha$ is an 
inverse-Ackermann function) in place of Dijkstra's algorithm leads to
static BC algorithms with the same time complexity as the corresponding
APSP algorithm, and both improve on Brandes.
Our incremental algorithm has better bounds than any of these APSP algorithms
on sparse graphs, and is no worse (and likely better) on dense graphs.
\subsection{Related work}
The notion of betweenness centrality 
was formalized by 
Freeman~\cite{Freeman77}  who also defined
other measures such as closeness centrality. 
Approximation algorithms and parallel algorithms for 
BC have been considered in \cite{BaderKMM07,GeisbergerSS08}, 
\cite{MadduriEJBC09} respectively. 
Lee~et~al.~\cite{Lee12}
present  a framework called QUBE (Quick Update of BEtweenness centrality) 
which allows edges to be inserted
and deleted from the graph, and  
recently, Singh~et~al.~\cite{SinghGIS13} build on the 
work of Lee~et~al.~\cite{Lee12} to allow
nodes to be added and deleted. 
Another recent work by Kas~et~al.\cite{KasWCC13}
deals with the problem of incremental BC and their algorithm is based on the 
dynamic all pairs shortest paths algorithm by Ramalingam and Reps~\cite{RR96}.
All the above mentioned papers give encouraging experimental results, 
but there are no performance guarantees.

Closest to our work is the incremental algorithm by 
Green~et~al.~\cite{GreenMB12} 
for 
a single edge insertion in
unweighted graphs, which maintains  a breadth-first search (BFS)
tree rooted at every source $s$ and identifies vertices for which 
the distance (BFS level)
from $s$ has changed. However, unlike our algorithm, they do not maintain 
the BFS DAG and hence
need to run a BFS for all vertices 
whose distances have changed.
Thus, in the worst case, their algorithm takes
$\Theta (mn + n^2)$ time  for compute the 
BC of all vertices in an {\em unweighted} graph.
In contrast, 
our algorithm, which takes time $O(m' \cdot n + n^2)$ 
even in the weighted case and even for a vertex update,
improves on the algorithm in \cite{GreenMB12} for unweighted graphs
when $m' = o(m)$.

\subsubsection{ Organization:}
Section~\ref{sec:prelims} presents
the notation we use,
and discusses Brandes' algorithm. Section~\ref{sec:increment} presents our
basic 
 BC algorithm for an incremental {\it edge} update, and this algorithm is
 extended  to a vertex update in
Section 4. Finally, in Section~\ref{sec:cache-obv} and Section~\ref{sec:static}
we give brief sketches of our cache-efficient and static algorithms.

\section{Preliminaries and background}
\label{sec:prelims}
We first consider directed graphs.
Let $G = (V, E)$ denote
a directed graph with positive real edge weights, given by 
$\weight: E \rightarrow \mathcal{R}^{+}$.
Let $\pi_{st}$ denote a path from $s$ to
$t$ in $G$. Define $\weight(\pi_{st}) = \sum_{e \in \pi_{st}} \weight(e)$ as the weight
of the path $\pi_{st}$. 
We use $d(s,t)$ to denote the weight of a shortest path from $s$ 
to $t$ in $G$, also called its \emph{distance}.

The following notation was developed by Brandes~\cite{Brandes01}.
For a source $s$ and a vertex $v$, let $P_s(v)$ denote the predecessors of $v$ on shortest
paths from $s$, i.e.,
\begin{eqnarray}
\label{eqn:pred-defn}
P_s(v) = \{ u \in V: (u, v) \in E \mbox{  and  } d(s, v) = d(s, u) + \weight(u, v) \}
\end{eqnarray}
Further, let $\sigma_{st}$ denote the number of shortest paths from 
$s$ to $t$ in $G$ (with $\sigma_{ss}=1$).
Finally, let $\sigma_{st} (v)$ denote the number of shortest paths from $s$ to $t$ in $G$ that pass through $v$.
It follows from the definition that,
\begin{eqnarray}
\label{eq:pair-dep}
\sigma_{st}(v)  =
\begin{cases}
   0   \ \ \ \ \ \ \ \ \ \ \ \   \mbox { if $d(s,t) < d(s, v) + d(v, t)$} \\
                 \sigma_{sv} \cdot \sigma_{vt}  \ \ \ \ \mbox{otherwise}
\end{cases}
\end{eqnarray}
%
 The dependency of the pair $s, t$ on an intermediate vertex $v$ is 
defined in~\cite{Brandes01} as the {\em pair dependency}
$\delta_{st}(v) = \frac{\sigma_{st}(v)} {\sigma_{st}}$.

\vspace{.05in}
For $v \in V$, the {\it betweenness centrality} $\bc(v)$ is defined by Freeman~\cite{Freeman77} as:
\begin{eqnarray}
\label{eqn:bc-defn}
\bc(v) = \sum_{s \neq v, t \neq v} \frac{\sigma_{st}(v)}{\sigma_{st}} = \sum_{s \neq v, t \neq v} \delta_{st}(v)
\end{eqnarray}
The following two-step procedure computes BC for all $v \in V$:
\begin{enumerate}%
\item For every pair $s, t \in V$, compute $\sigma_{st}$.
\item For every vertex $v \in V$, and for every $s, t$ pair, compute $\sigma_{st}(v)$ (Equation~\ref{eq:pair-dep}), and then compute $\bc(v)$ (Equation~\ref{eqn:bc-defn}).
\end{enumerate}
Step~1 above can be achieved by $n$ executions of Dijkstra's single source shortest paths
algorithm. 
Therefore Step~1 takes time $O(mn + n^2 \log{n})$ time, 
if we use a 
priority queue 
with $O(1)$ amortized cost for the decrease-key operation.
For every vertex $v$, Step~2
takes $O(n^2)$ time, since there are $O(n^2)$ pair dependencies. 
This gives a $\Theta(n^3)$ time algorithm to compute BC for all vertices.
Thus, the bottleneck of the above algorithm is the second step which explicitly sums up the pair
dependencies for every vertex.
To obtain a faster algorithm for sparse graphs, Brandes~\cite{Brandes01} defined the dependency of a vertex $s$ on a vertex $v$ as:
$\delta_{s \bullet} (v) = \sum_{t \in V \setminus {\{v, s\}}} \delta_{st}(v)$.
Brandes~\cite{Brandes01} also made the useful observation that the partial sums satisfy a recursive
relation. In particular, the dependency of a source $s$ on a vertex $v \in V$ can be written as:

\begin{eqnarray}
\label{eqn:rec-depend}
\delta_{s \bullet} (v) = \sum_{w: v \in P_s(w)} \frac{\sigma_{sv}}{\sigma_{sw}} \cdot \left ( 1 + \delta_{s \bullet}(w) \right)
\end{eqnarray}
(See \cite{Brandes01} for a proof of Equation~\ref{eqn:rec-depend}.) The above equation gives
an efficient algorithm for computing BC described in the next section.

\subsection{Brandes' algorithm}

We present a high level overview of Brandes' algorithm (this high-level
algorithm is given in Appendix~\ref{sec:Brandes-steps}).
The algorithm begins by
initializing the BC score for every vertex to $0$. Next,
for every $s \in V$, it executes Dijkstra's SSSP algorithm. During this step,
for every $t \in V$, it computes $\sigma_{st}$, the number of shortest paths from $s$ to $t$, and
$P_s(t)$, the set of predecessors of $t$ on shortest paths from
$s$. Additionally, the algorithm stores the vertices $v \in V$ in a stack $S$ in order of non-increasing value
of $d(s, v)$. Finally, to compute the BC score, 
the algorithm accumulates the dependency of $s$. 
We now elaborate on this final step, which is given in 
Algorithm~\ref{algo:accumulate} (Accumulate-dependency). 
This algorithm 
takes
as its input a source $s$ for which Dijkstra's SSSP algorithm has 
been executed and the stack $S$ containing
vertices ordered by distance from $s$. The algorithm repeatedly extracts a vertex from $S$ and 
accumulates the dependency 
using Equation \ref{eqn:rec-depend}.
The time taken by Algorithm~\ref{algo:accumulate} is linear in the size of the $\DAG$
rooted at $s$, i.e., it is $O(m^{\ast}_{s})$.

Note that in Brandes' algorithm, the set $S$ contains vertices $v \in V$
ordered in non-increasing value of $d(s,v)$. However, for the dependency accumulation it suffices 
that 
$S$ contains vertices $v \in V$ ordered in the reverse topological order of the $\DAG(s)$.
Such an ordering ensures that the dependency of a vertex $w$ is {\em accumulated} to any of its predecessor $v$,
only after all the successors of $w$ in $\DAG(s)$ have been processed. 
This observation is useful for our
incremental algorithm, since topological sort can be performed in linear time. 
\begin{algorithm}
\begin{algorithmic}[1]
\REQUIRE
For every $t \in V$: $\sigma_{st}, P_s(t)$ \\
\hspace{0.25in}A stack $S$ containing $v \in V$ in a suitable order (non-increasing $d(s,v)$ in \cite{Brandes01}) 
\STATE {\bf for} every $v \in V$ {\bf do}  $\delta_{s\bullet}(v) \leftarrow 0$ 
\WHILE {$S \neq \emptyset$}
\STATE  $w \leftarrow$ pop$(S)$
\STATE {\bf for} $v \in P_s(w)$ {\bf do} $\delta_{s \bullet}(v) \leftarrow \delta_{s \bullet}(v) + \frac{\sigma_{sv}}{\sigma_{sw}} \cdot \left( 1 + \delta_{s \bullet}(w) \right)$ 
\STATE {\bf if} $w \neq s$ {\bf then} $\bc(w) \leftarrow \bc(w) + \delta_{s \bullet}(w)$ 
\ENDWHILE
\end{algorithmic}
\caption{Accumulate-dependency($s,S$) (from \cite{Brandes01})}
\label{algo:accumulate}
\end{algorithm}

\section{Incremental edge update}
\label{sec:increment}

In this section we present our algorithm to recompute BC
of all vertices in a graph $G = (V, E)$  after an incremental edge update (i.e.,  adding a new edge $(u, v)$   or decreasing the edge weight of an existing edge $(u,v)$). We extend this to a
vertex update in the next section. We first consider directed graphs.

Let $G' = (V, E')$ denote the graph obtained after an edge update to ${G=(V,E)}$. Let $d(s,t), \sigma_{st}$, and $\delta_{s \bullet}(t)$ denote the distance from $s$ to $t$ in $G$,
the number of shortest paths from $s$ to $t$ in $G$, and the dependency of $s$ on $t$ in $G$ respectively, and  let $d'(s,t), \sigma'_{st}$, and $\delta'_{s \bullet}(t)$ denote these parameters in the graph $G'$. 

Our incremental algorithm relies on maintaining the SSSP DAG rooted at every $s \in V$ after an edge update. 
Let $\DAG(s)$, $\DAG'(s)$ denote the SSSP DAG rooted at each  $s$ in $G$ and in $G'$ respectively.
We show how to efficiently maintain these DAGs after an update. The updated DAGs give us the
updated $P'_s(t)$ for every $s, t \in V$. We also show how to maintain for every 
$s, t \in V$, $d'(s,t)$ and $\sigma'_{st}$. Then, using Algorithm~\ref{algo:accumulate},
we obtain the updated BC scores for all vertices.

We begin by making some useful observations. 
\begin{lemma} 
\label{lem:sets-equal}
Let $(u, v)$ denote the edge on which the weight is decreased. 
Then, for any vertex $x \in V$, the set of shortest paths from $x$ to $u$
is the same in $G$ and $G'$, and we have
\[
d'(x, u) =  d(x, u)  \mbox {     and    } d'(v, x) = d(v, x) ; ~~~ \sigma'_{xu} = \sigma_{xu} \mbox{   and   } \sigma'_{vx} = \sigma_{vx}  
\]
\end{lemma}
\begin{proof} 
Since edge weights are positive, the edge $(u,v)$ cannot lie on a shortest path to $u$ or from $v$. The lemma follows.
\qed
\end{proof}
By Lemma \ref{lem:sets-equal}, $\DAG(v) = \DAG'(v)$ after weight of $(u,v)$ is decreased.
%
The next lemma shows that after the weight of $(u,v)$ is decreased we 
can efficiently obtain the updated values $d'(s,t)$ and $\sigma'_{st}$ for any $s,t \in V$.
\begin{lemma} 
\label{lem:update-st}
Let the weight of edge $(u,v)$ be decreased to $\weight'(u, v)$, and for a
any given pair of vertices $s,t$, let
$D(s,t) = d(s, u) + \weight'(u, v) + d(v, t)$. Then,
\begin{enumerate}
\item If $d(s, t)  < D(s,t)$, then 
$d'(s, t) = d(s,t)$ and $\sigma'_{st} = \sigma_{st}$.
\item If $d(s, t)  = D(s, t)$, then $d'(s, t) = d(s,t)$ and 
$\sigma'_{st} = \sigma_{st} + (\sigma_{su} \cdot \sigma_{vt})$.
\item If $d(s, t)  > D(s,t)$, then
$d'(s, t) = D(s, t)$ and $\sigma'_{st} = \sigma_{su} \cdot \sigma_{vt}$.  
\end{enumerate}
\end{lemma}

\begin{proof}
Case 1 holds because the shortest path distance
 from $s$ to $t$ remains unchanged and no
new shortest path is created in this case.
In case 2, 
the shortest path distance from $s$ to $t$ remains unchanged, but there 
are $\sigma_{su} \cdot \sigma_{vt}$ new shortest paths from $s$ to $t$
created via edge $(u, v)$. 
In case 3, 
the shortest path distance from $s$ to $t$ decreases and all new
shortest paths pass through $(u,v)$.
\qed
\end{proof}

By Lemma \ref{lem:update-st}, the updated values $d'(s,t)$ and 
$\sigma'_{st}$  can be computed in constant time for each pair $s,t$.
Once we have the updated $d'(\cdot)$ and $\sigma'_{(\cdot)}$ values, we 
need the updated predecessors $P'_{s}(t)$ for every $s, t$ pair for
the betweenness centrality algorithm. In order to obtain these updated
predecessors efficiently,
we maintain the SSSP DAG rooted at every source $s\in V$.
The next section gives a very simple algorithm to maintain an 
SSSP DAG after an incremental edge update.

\subsection{Updating an SSSP DAG}
Let $\DAG(s)$, $\DAG'(s)$  denote the single source shortest path DAG 
rooted at $s$ in $G$ and $G'$ respectively.
Recall that since the update is a decrease in the edge weight of $(u, v)$, 
we know that $\DAG'(v) = \DAG(v)$.
For an $s, t$ pair we define a $flag(s,t)$ to capture the change in distance/number of shortest
paths from $s$ to $t$ after the decrease in the weight of the edge $(u,v)$
as follows.
\begin{equation}
flag(s,t) = 
\begin{cases}
\mbox{NUM-changed} \textrm{   \hspace{0.2in} if $d'(s,t) = d(s,t)$ and $\sigma'_{st} > \sigma_{st}$ } \\
\mbox{WT-changed} \textrm{   \hspace{0.3in} if $d'(s,t) < d(s,t)$} \\
\mbox{UN-changed} \textrm{   \hspace{0.32in} if $d'(s,t) = d(s,t)$ and $\sigma'_{st} = \sigma_{st}$ } \\
\end{cases}
\end{equation}

By Lemma \ref{lem:update-st}, $flag(s,t)$ can be computed in constant time
for each pair $s,t$.
On input $s$ and the updated edge $(u,v)$, 
Algorithm~\ref{algo:update-dag} (Update-DAG) constructs a set of edges 
$H$ using these $flag$ values,
together with $\DAG(s)$ and $\DAG(v)$.
We will show that $H$ contains exactly the edges in $\DAG'(s)$.
  We begin by initializing $H$ to empty (Step~\ref{dag-init}). We then consider each edge $(a, b)$ in $\DAG(s)$ (in Steps~2--4) and in $\DAG(v)$ (in Steps~5--7) and
depending on the value of  $flag(s,b)$ decide whether to include it in the set $H$. Finally, we check if the updated edge $(u,v)$ will be inserted in $H$ (in Steps~8--9).
It is clear from the algorithm that the time taken is linear in the size of $\DAG(s)$ and $\DAG(v)$
which is bounded by $O(m^{\ast}_s + m^{\ast}_v)$.

\begin{algorithm}
\begin{algorithmic}[1]
\REQUIRE $\textrm{DAG}(s)$, $\textrm{DAG}(v)$, and $flag(s,t), \forall t\in V$
\ENSURE An edge set $H$ after weight of edge $(u,v)$ has been decreased 
\STATE $H \leftarrow \emptyset$ \label{dag-init}
\FOR {each edge $(a,b) \in \DAG(s)$ and $(a,b) \neq (u,v)$} 
	\IF {$flag(s,b) = \mbox{UN-changed}$ or $flag(s,b) = \mbox{NUM-changed}$}
		\STATE $H \leftarrow H \cup \{(a,b)\}$
	\ENDIF
\ENDFOR
\FOR {each edge $(a,b) \in \textrm{DAG}(v)$}
	\IF {$flag(s,b) = \mbox{NUM-changed}$ or $flag(s,b) = \mbox{WT-changed}$}
		\STATE $H \leftarrow H \cup \{(a,b)\}$
	\ENDIF
\ENDFOR
\IF {$flag(s,v) = \mbox{NUM-changed}$ or $flag(s,v) = \mbox{WT-changed}$}
\STATE $H \leftarrow H \cup \{(u,v)\}$
\ENDIF
\end{algorithmic}
\caption{Update-DAG($s,(u,v)$)}
\label{algo:update-dag}
\end{algorithm}

\begin{lemma} \label{lem:dag-correct}
Let $H$ be the set of edges output by Algorithm~\ref{algo:update-dag}. 
An edge $(a,b) \in H$ if and only if $(a,b) \in \DAG'(s)$.
\end{lemma}

\begin{proof}
Since the update is an incremental update on edge $(u,v)$, we note that
for any $b$, a shortest path $\pi'_{sb}$ from $s$ to $b$ in $G'$ can be of two types: \\
(i) $\pi'_{sb}$ is a shortest path in $G$.
Therefore every edge on such a path is present in $\DAG(s)$ and 
each such edge is added to $H$ in Steps~2--4 of Algorithm~\ref{algo:update-dag}.\\
(ii) $\pi'_{sb}$ is not a shortest path in $G$. However, since 
$\pi'_{sb}$ is a shortest path in $G'$,
therefore $\pi'_{sb}$ is of the form $s \leadsto u \rightarrow v \leadsto b$.
Since shortest paths from $s$ to $u$ in $G$ and $G'$ are unchanged
(by Lemma \ref{lem:sets-equal}), the edges in the sub-path 
$s \leadsto u$ are
present in $\DAG(s)$ and they are added to $H$ in 
Steps~2--4 of Algorithm~\ref{algo:update-dag}. Finally, since shortest paths from $v$ to any $b$ 
in $G$ and $G'$ remain unchanged, the edges in the sub-path $v \leadsto b$ are present in $\DAG(v)$ and therefore
are added to $H$ in Steps~5--7 of Algorithm~\ref{algo:update-dag}.
\qed
\end{proof}
\subsection{Updating betweenness centrality scores}
In this section we present 
Algorithm~\ref{algo:update-bc} (Incremental-BC), which
updates the BC scores
for all vertices after an incremental edge update. 
The algorithm takes as input the graph $G = (V,E)$ and
the updated edge $(u,v)$, together with the current
values of  $d(s,t)$ and $\sigma_{st}$ for all $s,t \in V$,
and the current shortest path dag $\DAG(s)$, for every $s \in V$.
We begin by initializing the BC score to $0$ for every vertex (Step~\ref{update-bc-init}). 
For every pair $s,t$, we compute the updated $d'(s,t)$ and $\sigma'_{st}$ using Lemma~\ref{lem:update-st} (Step~\ref{update-bc-dist}).
Next, in Step~\ref{update-bc-dag} we compute  $\DAG'(s)$  
for every source $s \in V$, using Algorithm~\ref{algo:update-dag}.
Note that
this gives us the updated predecessor values $P'_s(t)$, for 
every $s, t \in V$. We perform a
topological sort of $\DAG'(s)$ to obtain an ordering of $v \in V$ and 
store the ordered vertices in a stack $S$.
Finally, using $S$, the $\sigma'_{st}$, and  the $P'_s(t)$, we 
run Brandes' accumulation of dependencies (Algorithm~\ref{algo:accumulate})
to compute the updated BC scores.

\begin{algorithm}
\begin{algorithmic}[1]
\REQUIRE updated edge $(u, v)$ with new weight $\weight'(u,v)$, \\
\hspace{0.2in} $d(s,t)$ and $\sigma_{st}$, $\forall \ s, t \in V$; $~~~ \DAG(s), \forall \ s \in V$ 
\ENSURE $\bc'(v)$, $\forall v\in V$ \\
\hspace{0.3in} $d'(s,t)$ and $\sigma'_{st}$ $\forall \ s, t \in V $; $~~~ \DAG'(s), \forall \ s \in V$ 
\STATE {\bf for} {every $v \in V$} {\bf do}	$\bc'(v) \leftarrow 0$   \label{update-bc-init}	
\STATE {\bf for} {every $s, t \in V$} {\bf do}  compute $d'(s,t), \sigma'_{st}$, $flag(s,t)$ \hspace{0.2in} // use Lemma \ref{lem:update-st}  \label{update-bc-dist}
\FOR {every $s \in V$}
	\STATE Update-DAG$(s,(u,v))$ \hspace{1.7in} // use Algorithm~\ref{algo:update-dag} \label{update-bc-dag}
	\STATE stack $S \leftarrow$ vertices in $V$ in a reverse topological order in $\DAG'(s)$ \label{update-bc-topo} 
	\STATE Accumulate-dependency$(s, S)$ \hspace{1.25in} // use Algorithm~\ref{algo:accumulate} \label{update-bc-accum}
\ENDFOR
\end{algorithmic}
\caption{Incremental-BC($G=(V,E)$)}
\label{algo:update-bc}
\end{algorithm}

The correctness of our algorithm follows from the correctness of maintaining $d(\cdot), \sigma_{(\cdot)}$, and the updated DAGs. We now bound
the time complexity of our algorithm. In 
Step~\ref{update-bc-dist} of Algorithm~\ref{algo:update-bc} we spend  
constant time to compute the updated $d'(s,t)$ and $\sigma'_{st}$ values
for an $s,t$ pair, hence $O(n^2)$ time for all pairs.
In Step~\ref{update-bc-dag}, we obtain $\DAG'(s)$ for each $s\in V$ using Algorithm~\ref{algo:update-dag}, which takes time $O(m^{\ast}_s + m^{\ast}_v)$ for source $s$.
 This amounts to a total of 
$O(m^{\ast}_v n + \sum_{x \in V} m^{\ast}_x) = O( (\mVmax + m^{\ast}_v) n)$.
Finally, Steps~\ref{update-bc-topo} and \ref{update-bc-accum} take time 
linear in the size of the updated DAG. Thus the time complexity of our 
Incremental-BC  algorithm for an edge update is bounded by 
$O( (\mVmax + m^{\ast}_v) n + n^2)$. This establishes Theorem
\ref{thm:main1-intro} for directed graphs for an update on a single edge.

We remark that both Update-DAG and Increment-BC use very 
simple data structures: arrays, lists, and stacks.
Thus our algorithms are very simple to implement, and should have small
constant factors in their running time.

\paragraph{\bf Undirected graphs:}
Consider an undirected positive edge weighted graph $G = (V, E)$. 
To obtain an incremental algorithm
for $G$, we first construct the corresponding directed graph 
$G_D = (V, E_D)$
from $G$ by replacing every edge $\{a, b\} \in E$ by two 
directed edges $(a, b)$ and $(b,a)$, both with 
the weight of the undirected edge $\{a,b\}$. 
Since edge weights are positive, no shortest path can contain a cycle,
hence the SSSP sub-graph rooted at $s$ in $G$ is the same as the SSSP DAG
rooted at $s$ in $G_D$. 
Our incremental algorithm for $G$ for an edge update on $\{u, v\}$ 
is now simple: we apply two incremental edge
updates on $G_D$, one on $(u, v)$ and another on $(v, u)$. 
Since the number of edges in $G_D$ is only twice
the number of edges in $G$, all other parameters also increase by only 
a constant factor. This establishes Theorem \ref{thm:main1-intro} for edge updates in 
undirected graphs.

\section{Incremental vertex update}
\label{sec:vert-update}
Here we show how to extend the incremental edge update algorithm to handle 
an incremental update to a vertex $v$ in $G = (V, E)$, where we allow
an incremental edge update on any subset of edges incoming to and outgoing from the vertex $v$.
Our algorithm 
is a natural extension of the algorithm for single edge update. 
As in the algorithm for edge update, for every $s, t \in V$, we maintain
$d(s, t), \sigma_{st}$, and for every $s \in V$, $\DAG(s)$, the SSSP DAG rooted at $s$ in $G$. 
Once we have the updated DAGs, we use a topological sort and accumulate dependencies in the reverse topological
order to get updated BC scores for all vertices.
However, instead of working only with the graph $G$, 
here we also  work with 
the graph $G_R = (V, E_R)$, which is obtained by reversing every edge in $G$. That is $(a, b) \in E_R \textrm { iff } (b, a) \in E$.
Thus, for every $s \in V$, we also maintain $\DAG_R(s)$, the SSSP DAG rooted at $s$ in $G_R$. 

Broadly, our vertex update on $v$ is processed as follows.
Let $\Ein(v)$ and $\Eout(v)$ denote the set of  updated edges incoming to $v$ and outgoing from $v$ respectively. 
We process 
$\Ein(v)$ and $\Eout(v)$ in two steps.
Let $G'$ denote the graph
obtained by applying to $G$ the
updates on edges in $\Ein(v)$. Let $G''$ denote the
graph obtained by applying to $G'$ updates on edges in $\Eout(v)$.
For $s, t \in V$, let $d(s,t), d'(s,t)$ and $d''(s,t)$ denote the distance from $s$ to $t$ in $G, G'$, and $G''$ respectively.
We use a similar notation for other terms, such as $\sigma_{st}$ and $\DAG(s)$.
The main observation that we use is the following: 
Since $\Ein(v)$ contains updated edges incoming to $v$, the SSSP DAGs rooted at $v$ in $G$ and
in $G'$ are the same, that is, $\DAG(v) = \DAG'(v)$. We show that
 Algorithm~\ref{algo:update-dag} is readily adapted to obtain $\DAG'(s)$, for all $s\in V$.
At this point, our goal is to apply the updates in $\Eout(v)$ to $G'$ and obtain 
$\DAG''(s)$ for every $s$. To achieve this efficiently, we first 
obtain $\DAG'_R(s)$ for every $s$.
The reason to work with the reverse graph $\Grev'$ and DAGs in the reverse graph is that 
the edges in $\Eout(v)$ are in fact incoming edges to $v$ in $G'_R$. Hence our
method to maintain DAGs when incoming edges are updated works as it is on $G'_R$
and we obtain $\DAG''_R(s)$, for every $s$. Finally, using $\DAG''_R(s)$ for every $s$
we efficiently build $\DAG''(s)$.

We now give details of each step of our algorithm starting with the graph $G$ till
we obtain the $\DAG'_R(s)$ for every $s$.  
\begin{enumerate}
\item[(A)] {\bf Compute $d'(s, v)$ and $\sigma'_{sv}$ for any $s$:}
In this step we show how to compute in $G'$ the distance and number of shortest paths to $v$ from any $s$. We make
the following definitions:
For $(u_j, v) \in \Ein(v)$, let
$D_j(s, v) = d(s, u_j) + \weight'(u_j, v)$.
Since the updates on edges in $\Ein(v)$ are incremental, it follows that
\begin{eqnarray}
d'(s, v) = \min \{ d(s, v), \min_{j: (u_j, v) \in \Ein(v)} \{D_j(s, v) \} \}
\end{eqnarray}
Further, if $d'(s, v) = d(s, v)$, we define
\begin{eqnarray}
\newpaths  = | \{\pi'_{sv} : \mbox {$\pi'_{sv}$ is a shortest path in $G'$ and $\pi'_{sv}$ uses $e \in \Ein(v)$} \}|
\end{eqnarray}
We also need to compute $\sigma'_{sv}$, the number of shortest paths from $s$ to $v$ in $G'$.
It is straightforward to compute $d'(s,v)$, $\sigma'_{sv}$, and $\newpaths$ in $O(|\Ein(v)|)$ time. 
Algorithm~\ref{algo:compute-dist-vert} gives the details of this step.
\begin{figure}
\begin{minipage}{0.5\linewidth}
\begin{algorithm} [H]
\begin{algorithmic}[1]
\REQUIRE $\Ein(v)$ with updated weights $\weight'$ \\
\hspace{0.2in} $d(s,t)$ and $\sigma_{st}$, $\forall \ s, t \in V$
\ENSURE $d'(s,v), \sigma'_{sv}, \newpaths$
\STATE $\hat{\sigma}'_{sv} \leftarrow 0$, $\sigma'_{sv} \leftarrow \sigma_{sv}$, $currdist \leftarrow d(s,v)$ 
\FOR {each edge $(u_i,v) \in \Ein(v)$}
	\IF {$currdist = d(s,u_i) + \weight'(u_i,v)$}
	\STATE $\sigma'_{sv} \leftarrow \sigma'_{sv} + \sigma_{su_i}$
	\STATE $\hat{\sigma}'_{sv} \leftarrow \hat{\sigma}'_{sv} + \sigma_{su_i}$ \label{alg4:updnewpath}
	\ELSIF {$ currdist > d(s,u_i) + \weight'(u_i,v)$}
	\STATE $currdist \leftarrow d(s,u_i) + \weight'(u_i,v)$
	\STATE $\sigma'_{sv} \leftarrow \sigma_{su_i}$		
	\ENDIF	
\ENDFOR
\STATE $d'(s,v) \leftarrow currdist$
\end{algorithmic}
\caption{Compute-Dist-to-v $(s, \Ein(v))$}
\label{algo:compute-dist-vert}
\end{algorithm}
\end{minipage}
\begin{minipage}{0.5\linewidth}
\vspace{-0.18in}
\begin{algorithm}[H]
\begin{algorithmic}[1]
\REQUIRE $\DAG_R(s)$, and $R_t, flag(s,t), \forall t\in V$
\ENSURE An edge set $X$ after update on edges in $\Ein(v)$
\STATE $X \leftarrow \emptyset$
\FOR {each edge $(a,b) \in \DAG_R(s)$} \label{rdag:for1start}
	\IF {$flag(b,s) = \mbox{UN-changed}$ or $flag(b,s) = \mbox{NUM-changed}$}
	\STATE $X \leftarrow X \cup (a,b)$ \label{rdag:step1}
	\ENDIF
\ENDFOR
\FOR {each $b \in V \setminus \{s\}$} \label{rdag:for2start}
	\IF {$flag(b,s) = \mbox{NUM-changed}$ or $flag(b,s) = \mbox{WT-changed}$}
	\STATE $X \leftarrow X \cup R_b$ \label{rdag:step2}
	\ENDIF
\ENDFOR 
\end{algorithmic}
\caption{Update-Reverse-DAG($s$, $\Ein(v)$)}
\label{algo:update-rdag}
\end{algorithm}

\end{minipage}
\end{figure}

\item[(B)] {\bf Compute $d'(s,t)$ and $\sigma'(s, t)$ for all $s, t$:} Assuming that we have computed $d'(s, v), \sigma'_{sv}$ 
and $\hat{\sigma}'_{sv}$, we show that the values $d'(s, t)$ and $\sigma'(s, t)$ can be computed efficiently. We state
Lemma~\ref{lem:update-st-vert} which captures this computation. This lemma is similar to Lemma~\ref{lem:update-st} in the edge update case.
\begin{lemma} 
\label{lem:update-st-vert}
Let $\Ein(v)$ denote the set of edges incoming to $v$ which have been updated. 
Let $G'$ denote the graph obtained by applying update in $\Ein(v)$ to $G$.
For any $s \in V$ and $t \in V \setminus \{v\}$, let $D(s,t) = d'(s,v) + d(v,t)$.
\begin{enumerate}
\item If $d(s, t)  < D(s,t)$, then 
$d'(s, t) = d(s,t)$ and $\sigma'_{st} = \sigma_{st}$.
\item If $d(s, t)  = D(s, t)$ and $d(s,v) = d'(s,v)$, then $d'(s, t) = d(s,t)$ and 
${\sigma'_{st} = \sigma_{st} + \newpaths \cdot \sigma_{vt}}$.
\item If $d(s, t)  = D(s, t)$ and $d(s,v) > d'(s,v)$, then $d'(s, t) = d(s,t)$ and
${\sigma'_{st} = \sigma_{st} + \sigma'_{sv} \cdot \sigma_{vt}}$.
\item If $d(s, t)  > D(s,t)$, then
$d'(s, t) = D(s, t)$ and $\sigma'_{st} = \sigma'_{sv} \cdot \sigma_{vt}$.  
\end{enumerate}
\end{lemma}

\item[(C)] {\bf Compute $\DAG'(s)$ for every $s$:} Assuming that we have computed $d'(s,t)$ and $\sigma'(s,t)$ for
all $s, t \in V$, we now show how to obtain the updated DAGs. Note that we can readily compute
the value $flag(s,t)$ for every $s, t$, using the updated distances and number of shortest paths.
The algorithm to compute $\DAG'(s)$ for any $s \in V$ is very similar to Algorithm~\ref{algo:update-dag} in
the edge update case. The only modification to  Algorithm~\ref{algo:update-dag} we need is in
 Steps 8--9 where we consider a single edge $(u, v)$ -- instead in this case 
we consider every edge in $\Ein(v)$.

\item[(D)]{\bf Compute $\DAG'_R(s)$ for every $s$:} 
We now need to update $\DAG_R(s)$, for every $s$, for which we use Algorithm~\ref{algo:update-rdag}.
This algorithm requires for every $t \in V$, a set $R_t$ which is defined as follow:
\begin{eqnarray}
{R_t =  \{(a,t) \ | \ (t,a) \in \DAG'(t) \textrm{ and } \weight'(t,a) + d'(a,v) = d'(t,v) \}}
\end{eqnarray}
The set $R_t$ is the set of (reversed) outgoing edges from $t$ in $\DAG'(t)$ that lie on a shortest path from $t$ to $v$ in $G'$.

Consider an
 edge $e=(a,b)$ in the updated $\DAG'_R(s)$.  If $e$ is in $\DAG_R(s)$, it is added to $\DAG'_R(s)$ by Steps \ref{rdag:for1start}--\ref{rdag:step1}. If $e$ lies on a new shortest path present only in $\Grev'$,
 its reverse must also  lie on a shortest path that goes through $v$ in $G'$, and it will be added to $\DAG'_R(s)$ by the $R_b$ during Steps \ref{rdag:for2start}--\ref{rdag:step2}. ($R_b$ could also contain edges on old shortest paths through $v$  already processed in Steps \ref{rdag:for1start}--\ref{rdag:step1}, but even in that case each edge 
 is added to $\DAG'_R(s)$ at most twice by Algorithm \ref{algo:update-rdag}.)
Note that we do not need to process edges $(u_j,v)$
in $E_i$ separately (as with edge $(u,v)$ in Algorithm 2),
because these edges will be present in the relevant $R_{u_j}$.
The correctness of Algorithm~\ref{algo:update-rdag} follows from 
Lemma~\ref{lem:update-rdag}. This lemma similar to Lemma~3  and its proof
is in the Appendix.
\begin{lemma}
\label{lem:update-rdag}
An edge $(a,b) \in X$ if and only if $(a,b) \in \DAG'_R(s)$ after the incremental update of the set $\Ein(v)$.
\end{lemma} 
\end{enumerate}

At this point, after having Steps~(A)--(D) above executed, we have processed the updates in $\Ein(v)$ and obtained the modified distances $d'(\cdot)$, modified counts $\sigma'_{(\cdot)}$ and $\DAG'(s)$
for every $s$. In addition, we have obtained the modified reverse DAG for every $s$. To process the updates in $\Eout(v)$,
we work with $G_R'$. Since we are processing incoming edges in $G_R'$, our earlier steps apply 
unchanged, and we
obtain modified values for $d''(\cdot)$, $\sigma''_{(\cdot)}$, and $\DAG''_R(s)$ for every $s$. Finally,
using Algorithm~\ref{algo:update-rdag} we obtain the $\DAG''(s)$ for every $s$.
This completes our vertex update and to compute the updated BC values, we apply Brandes' accumulation technique (Algorithm~\ref{algo:accumulate}).

\paragraph{Complexity:} 
Computing $d'(s,v),\sigma'_{sv}$ and $\hat{\sigma}'_{sv}$  requires time ${O(|\Ein(v)|) = O(n)}$ for each $s$, and hence $O(n^2)$ time 
for all sources.
Applying Lemma \ref{lem:update-st-vert} 
to all pairs of vertices takes
time $O(n^2)$. 
For any $s$, the complexity of modified Algorithm~\ref{algo:update-dag} becomes $O(m_s^{*}+  \mStar_v +n)$ which
leads to a total of $O(m' \cdot n + n^2)$.

Creating a set $R_t$ 
requires at most 
${O(E^*\cap \{\textrm{outgoing edges of }t\})}$, 
so the overall complexity for all 
the sets is $O(m^*)$. Finally, we bound the complexity 
of Algorithm \ref{algo:update-rdag}: 
the algorithm adds $(a,b)$ in a 
reverse DAG edge set $X$ at most twice. 
Since $\sum_{s \in V}{|E(\DAG'(s))|} = \sum_{s \in V}{|E(\DAG'_R(s))|}$, 
at most ${O(\mVmax n)}$ edges can be inserted into all the sets $X$ 
when Algorithm~\ref{algo:update-rdag} is executed
over all sources. Finally, since applying the updates 
in $\Eout(v)$  requires a symmetric 
procedure starting from the reverse DAGs, the final complexity bound of $O(m'~\cdot~n~+~n^2)$ follows.

\section{Cache-oblivious Implementation}
\label{sec:cache-obv}
Our incremental algorithms have efficient cache-oblivious implementaions.
We  describe here an efficient  cache-oblivious implementation of 
Algorithm \ref{algo:update-bc} (Incremental-BC). 
As noted in Section~\ref{sec:intro}, 
$scan(r) = r/B$ and $sort(r) = (r/B) \log_M r$ are desirable cache-efficient
bounds for a computation of size $r$.

Let $V=\{1, \ldots, n\}$.
Further, assume that for an $s, t$ pair, we have the values 
$d(s,t), \sigma(s,t), flag(s,t)$ stored in
an array $\A[1..n^2]$ ordered by the first component
$s$ and then by the second component $t$ (i.e.,  stored as an $n\times n$ 
row major 
matrix).

To perform Step 2 of Algorithm \ref{algo:update-bc} using a scan, we extract 
the subarray $\A_{v \bullet}$
containing the entries with $v$ as the first component, and 
the subarray $\A_{\bullet u}$ containing the entries with $u$ as the second component.
We scan the three arrays $\A, \A_{v \bullet}$, and $\A_{\bullet u}$ 
to compute the updated values $d', \sigma'$,
and $flag$ for each pair $s,t$ in the order they appear in $\A$. 
These are stored in $\A'[1..n^2]$.
Overall, Step 2 incurs $O(scan(n^2))$ cache misses.
We execute the {\bf for} loop in Step~3 in increasing order of $s\in V$,
and hence we will access successive segments of $\A'$.

For Step 4, we sort the edges $(a,b)$ in $\DAG(v)$ and in
$\DAG(s)$ in nondecreasing order  
of $b$. Across all $s\in V$ this can be performed in
$n \cdot sort(\mVmax)$ cache misses.
Then, Algorithm \ref{algo:update-dag} can be executed in $scan(n^2) + n \cdot scan(\mStarV)$ cache misses
across all sources since we only need $flag(s,b)$ when we examine
edge $(a,b)$ in $\DAG(s)$.

Instead of the
reverse topologically sorted order for the stack $S$ in Step 5
we use nonincreasing order of $d'(s,t)$, $t\in V$ (as in the Brandes
static algorithm). This is computed in $sort(m_s^*)$ cache misses for
source $s$ and $n \cdot sort(\mVmax)$ across all sources.  

In the final step, Step~6, we execute Algorithm \ref{algo:accumulate}. 
As input to this algorithm, for a given $s$, we need to generate for 
every $w \in V$, the predecessor list $P_s(w)$. For this,
for every $v \in P_s(w)$, we store
the value $d'(s,v)$. This computation is done by sorting the edges $(a, b) \in DAG(s)$
by the first component $s$. This sorted list will be a subsequence of the
row for $s$ in $\A'[1..n^2]$,
and $d'(s,v)$ can be copied to each edge $(v, w)$ in $scan(m_s^{*})$ cache misses.
The $P_s(v)$ lists are then generated with another sort, and each entry
in the predecessor list will contain the associated $d'$ value.
Over all sources $s$, this computation takes $O(scan(n^2) + n \cdot sort(\mVmax))$ 
cache misses. 

Having generated the predecessor lists, we need to execute 
Algorithm \ref{algo:accumulate} for each source $s$.
The
cache-oblivious implementation of Algorithm \ref{algo:accumulate} 
is somewhat different from the earlier pseudocode. We use
an optimal cache-oblivious max-priority queue $Z$ \cite{ArgeBDHM07}, that is
initially empty.
Each element in $Z$ has an ordered pair $(d'(s,v), v)$ as its key value, 
and also has auxiliary data as described below.
Consider the execution of Step~4 in Algorithm \ref{algo:accumulate} for vertices
$v\in P_s(w)$. 
Instead of computing the contribution of $w$ to $\delta_{s\bullet}(v)$ for each
$v\in P_s(w)$ when $w$ is processed, we insert an element into $Z$ with
key value $(d'(s,v),v)$ and auxiliary data 
$(w,\sigma_{sw}, \delta_{s\bullet}(w))$.
With this scheme, entries will be extracted from $Z$ in nonincreasing values of 
$d'(s,v)$, and all entries for a given $v$ will be extracted consecutively.
We compute $\delta_{s\bullet}(v)$ as these extractions occur from $Z$. Note
that the stack $S$ is needed only 
to identify the sinks in $\DAG(s)$ (which will have no entry in $Z$).
So, we can dispense with $S$, and instead initially insert an
element with key value $(d'(s,t),t)$ and NIL auxiliary data for each
sink $t$ in $\DAG(s)$.
Using~\cite{ArgeBDHM07}, Step 6 takes $sort(m_s^*)$ cache misses
for source $s$, hence over all sources,
Step 6 takes $O(n \cdot sort (\mVmax))$.


Thus, the overall 
cache-oblivious incremental algorithm for betweenness centrality incurs
$O(scan(n^2) +  n \cdot sort(m'))$ cache misses, where $m'= \mVmax + m^*_v$.
This is a significant 
improvement over any algorithm that computes shortest paths, as is the case
with the static Brandes algorithm,  since all of 
the known algorithms for shortest paths 
face the bottleneck of unstructured accesses to adjacency lists 
(see, e.g., \cite{MehlhornM02}).

\section{Static betweenness centrality}
\label{sec:static}
In this section we present static algorithms that compute betweenness 
centrality faster than the Brandes algorithm.
We first consider an algorithm that is based on the Hidden Paths algorithm by
Karger et al.~\cite{KKP93} together with Brandes' accumulation technique. 
The Hidden Path algorithm runs Dijkstra's 
SSSP in parallel from each vertex. It identifies all pairs shortest paths while
only examining the edges in $E^*$, the set of edges that actually lie
on some shortest path.
A similar algorithm with the same running time of 
$O(\mStar n+n^2\log n)$ was developed independently by McGeoch~\cite{McGeoch95}; here $m^* = |E^*|$.

Our Static-BC algorithm is presented as Algorithm~\ref{algo:static}.
In Step 1 we run the Hidden Paths algorithm to compute $E^*$ as well as the 
shortest path distances for every pair of vertices.
This is the step with the dominant cost, 
while in Steps 2--8 the complexity is strictly related to the size of $E^*$. In Steps 2--4, we identify the edges in each shortest path DAG and in each $P_s(v)$: for every edge $(u,v) \in E^{\ast}$, 
if ${d(s,u)+\weight(u,v) = d(s,v)}$, then we add the edge $(u,v)$ 
to $\DAG(s)$ and the vertex $u$ to $P_s(v)$.
The overall time spent for constructing the DAGs and the predecessor lists is bounded by $O( \mStar n )$.
Step~\ref{static-comp-dist} counts the number of shortest paths from $s$ to $v$ for all $v \in V$,  
by traversing $\DAG(s)$ according to the topological order of its vertices, maintained in the double-ended queue $Q$ (created in Step 6) used as a queue. We accumulate the path counts for a vertex $v$ according to the formula $\sigma_{sv} = \sum_{(u,v) \in \DAG(s)}\sigma_{su}$.
This takes time linear in the size of $\DAG(s)$. Therefore, across all sources, we spend time which is bounded by 
$O(n^2 + \sum_{s \in V}{m^{\ast}_s}) = O(\mVmax n + n^2)$.
Finally in 
Step 8, using $Q$ as a stack (reverse topological order), we call Accumulate-dependency($s,Q$) (Algorithm~\ref{algo:accumulate}) to 
accumulate dependencies. Thus the overall running time of this static BC
algorithm is $O(m^*n + n^2 \log n)$.
The correctness of the algorithm follows from the correctness of
the Hidden Paths algorithm and the Brandes' accumulation technique.

\begin{algorithm}
\begin{algorithmic}[1]
\STATE Using the Hidden Paths algorithm, compute $E^{\ast}$, and $d(s,t)$ for every $s,t \in V$
\FOR {each node $s \in V$}
\FOR {each $(u,v) \in E^*$}
\STATE {\bf if} $d(s,u)+\weight(u,v)=d(s,v)$ {\bf then} add $(u,v)$ to $\DAG(s)$ and $u$ to $P_s(v)$
\ENDFOR
\ENDFOR
\FOR {each DAG$(s)$}
\STATE compute a dequeue $Q$ containing the nodes of $\DAG(s)$ in topological order 
\STATE for all $v\in V$, compute  $\sigma_{sv} = \sum_{(u,v) \in \DAG(s)}\sigma_{su}$ by accumulating path counts on vertices extracted from $Q$ in queue order (topological order)\label{static-comp-dist}
\STATE Accumulate-dependency($s,Q$), using $Q$ in stack order \hspace{.10in} //use Algorithm~\ref{algo:accumulate}
\ENDFOR 
\end{algorithmic}
\caption{Static-BC($G=(V,E)$)} 
\label{algo:static}
\end{algorithm}

Algorithm \ref{algo:static} can be expected to run faster than
Brandes' on many graphs since $m^*$ is often much smaller than $m$.
However, its worst-case running time is asymptotically the same as Brandes'.
We observe that if we replace the Hidden Paths algorithm in Step 1 of
Algorithm \ref{algo:static} with any other APSP algorithm that identifies
a set of edges $E'\supseteq E^*$, we can use $E'$ in place of $E$ in Step 3,
and obtain a correct static BC algorithm that runs in time 
${O(m'n+n^2 + T')}$, where $m' = |E'|$ and $T'$ is the running time of
the APSP algorithm used in Step 1.
In particular, if we use
one of the faster APSP
 algorithms for positive real-weighted graphs
(Pettie~\cite{Pettie04} for directed graphs or \cite{PR05} for
 undirected graphs) in Step 1, we can obtain asymptotically 
faster BC algorithms than Brandes' by using $E'=E$.
With Pettie's algorithm~\cite{Pettie04}, we obtain an $O(mn + n^2 \log\log n)$
time algorithm for static betweenness centrality in directed graphs, and 
with the algorithm of Pettie and Ramachandran~\cite{PR05}, we obtain a
static betweenness centrality algorithm for undirected graphs that runs in
$O(mn \cdot \log \alpha(m,n))$, where $\alpha$ is an inverse-Ackermann
function.

\bibliographystyle{abbrv}
\bibliography{references}

\begin{thebibliography}{10}

\bibitem{ArgeBDHM07}
L.~Arge, M.~A. Bender, E.~D. Demaine, B.~Holland-Minkley, and J.~I. Munro.
\newblock An optimal cache-oblivious priority queue and its application to
  graph algorithms.
\newblock {\em SIAM J. Comput.}, 36(6):1672--1695, 2007.

\bibitem{BaderKMM07}
D.~A. Bader, S.~Kintali, K.~Madduri, and M.~Mihail.
\newblock Approximating betweenness centrality.
\newblock In {\em Proc. 5th WAW}, pages 124--137, 2007.

\bibitem{Brandes01}
U.~Brandes.
\newblock A faster algorithm for betweenness centrality.
\newblock {\em Journal of Mathematical Sociology}, 25(2):163--177, 2001.

\bibitem{Coffman}
T.~Coffman, S.~Greenblatt, and S.~Marcus.
\newblock Graph-based technologies for intelligence analysis.
\newblock {\em Commun. ACM}, 47(3):45--47, 2004.

\bibitem{Freeman77}
L.~C. Freeman.
\newblock A set of measures of centrality based on betweenness.
\newblock {\em Sociometry}, 40(1):35--41, 1977.

\bibitem{FG85}
A.~Frieze and G.~Grimmett.
\newblock The shortest-path problem for graphs with random arc-lengths.
\newblock {\em Discrete Applied Mathematics}, 10(1):57 -- 77, 1985.

\bibitem{GeisbergerSS08}
R.~Geisberger, P.~Sanders, and D.~Schultes.
\newblock Better approximation of betweenness centrality.
\newblock In {\em Proceedings of 10th ALENEX}, pages 90--100, 2008.

\bibitem{GreenMB12}
O.~Green, R.~McColl, and D.~A. Bader.
\newblock A fast algorithm for streaming betweenness centrality.
\newblock In {\em Proceedings of 4th PASSAT}, pages 11--20, 2012.

\bibitem{HZ85}
R.~Hassin and E.~Zemel.
\newblock On shortest paths in graphs with random weights.
\newblock {\em Mathematics of Operations Research}, 10(4):557 -- 564, 1985.

\bibitem{KKP93}
D.~R. Karger, D.~Koller, and S.~J. Phillips.
\newblock Finding the hidden path: Time bounds for all-pairs shortest paths.
\newblock {\em SIAM J. Comput.}, 22(6):1199--1217, 1993.

\bibitem{KasWCC13}
M.~Kas, M.~Wachs, K.~M. Carley, and L.~R. Carley.
\newblock Incremental algorithm for updating betweenness centrality in
  dynamically growing networks.
\newblock In {\em Proceedings of 2013 IEEE/ACM International Conference on
  Advances in Social Networks Analysis and Mining}, page to appear, 2013.

\bibitem{KA12}
N.~Kourtellis, T.~Alahakoon, R.~Simha, A.~Iamnitchi, and R.~Tripathi.
\newblock Identifying high betweenness centrality nodes in large social
  networks.
\newblock {\em Social Network Analysis and Mining}, pages 1--16, 2012.

\bibitem{Krebs02}
V.~Krebs.
\newblock Mapping networks of terrorist cells.
\newblock {\em CONNECTIONS}, 24(3):43--52, 2002.

\bibitem{Lee12}
M.-J. Lee, J.~Lee, J.~Y. Park, R.~H. Choi, and C.-W. Chung.
\newblock Qube: a quick algorithm for updating betweenness centrality.
\newblock In {\em Proc. 21st WWW Conference}, pages 351--360, 2012.

\bibitem{LRP89}
M.~Luby and P.~Ragde.
\newblock A bidirectional shortest-path algorithm with good average-case
  behavior.
\newblock {\em Algorithmica}, 4(1-4):551--567, 1989.

\bibitem{MadduriEJBC09}
K.~Madduri, D.~Ediger, K.~Jiang, D.~A. Bader, and D.~G. Chavarr\'{\i}a-Miranda.
\newblock A faster parallel algorithm and efficient multithreaded
  implementations for evaluating betweenness centrality on massive datasets.
\newblock In {\em IPDPS}, pages 1--8. IEEE, 2009.

\bibitem{McGeoch95}
C.~C. McGeoch.
\newblock All-pairs shortest paths and the essential subgraph.
\newblock {\em Algorithmica}, 13(5):426--441, 1995.

\bibitem{MehlhornM02}
K.~Mehlhorn and U.~Meyer.
\newblock External-memory breadth-first search with sublinear i/o.
\newblock In {\em ESA}, volume 2461, pages 723--735, 2002.

\bibitem{Pettie04}
S.~Pettie.
\newblock A new approach to all-pairs shortest paths on real-weighted graphs.
\newblock {\em Theoretical Computer Science}, 312(1):47 -- 74, 2004.

\bibitem{PR05}
S.~Pettie and V.~Ramachandran.
\newblock A shortest path algorithm for real-weighted undirected graphs.
\newblock {\em SIAM J. Comput.}, 34(6):1398--1431, 2005.

\bibitem{PCW05}
J.~W. Pinney, G.~A. McConkey, and D.~R. Westhead.
\newblock Decomposition of biological networks using betweenness centrality.
\newblock In {\em Proceedings of 9th International Conference on Research in
  Computational Molecular Biology}, 2005.

\bibitem{RR96}
G.~Ramalingam and T.~W. Reps.
\newblock On the computational complexity of dynamic graph problems.
\newblock {\em Theor. Comput. Sci.}, 158(1{\&}2):233--277, 1996.

\bibitem{SinghGIS13}
R.~R. Singh, K.~Goel, S.~Iyengar, and Sukrit.
\newblock A faster algorithm to update betweenness centrality after node
  alteration.
\newblock In {\em Proc. 10th WAW, to appear}, 2013.

\end{thebibliography}

\begin{appendix}

\section{Brandes' algorithm}
For completeness we present the Brandes algorithm here.

\label{sec:Brandes-steps}
\begin{algorithm}
\begin{algorithmic}[1]
\STATE {\bf for} every  $v \in V$ {\bf do} $\bc(v) \leftarrow 0$ \label{brandes-init}
\FOR {every $s \in V$}
        \STATE run Dijkstra's SSSP from $s$ and compute $\sigma_{st}$ and $P_s(t), 
 \forall \ t \in V \setminus \{s\}$ \label{brandes-dijkstras}
        \STATE store the explored nodes in a stack $S$ in non-increasing distance from $s$
        \STATE accumulate dependency of $s$ on all $t \in V \setminus {s}$ using Algorithm~\ref{algo:accumulate} \label{brandes-accum}
\ENDFOR
\end{algorithmic}
\caption{Betweenness-centrality($G=(V,E)$) (from \cite{Brandes01})}
\label{algo:brandes}
\end{algorithm}


\section{Proof of Lemma~\ref{lem:update-rdag}}
Suppose $(a,b) \in \DAG'_R(s)$. This implies that there is at least one shortest from $s$ to $b$ in $\Grev'$ that uses
the edge $(a,b)$. And similarly a shortest path from $b$ to $s$ that uses the edge $(b,a)$ in $G'$. To show that $(a,b) \in X$ we consider the possible values for $flag(b,s)$:
\begin{enumerate}
\item If $flag(b,s) = \mbox{UN-changed}$ then the set of shortest paths from $b$ to $s$ in $G$ and $G'$ are
the same, so there is a shortest path from $b$ to $s$ that uses the edge $(b,a)$ in $G$. Further, since $\DAG_R(s)$ is correctly kept, the edge $(a,b)$ is added to $X$ from Step~\ref{rdag:step1}.
\item If $flag(b,s) = \mbox{WT-changed}$ then every shortest path from $b$ to $s$ in $G'$ goes through the updated vertex $v$. This implies that any shortest path $\pi'_{bs} \in G'$ is of the form $\pi'_{bs} = b \leadsto v \leadsto s$. Similarly every shortest path from $s$ to $b$ in $\Grev'$ goes through $v$ and, since $(a,b) \in \DAG'_R(s)$, at least one shortest path in $\Grev'$ is of the form $s \leadsto v \leadsto a \rightarrow b$. Therefore a shortest path of the form $b \rightarrow a \leadsto v \leadsto s$ is in $G'$, and the edge $(b,a)$ must be one of the outgoing edges of $b$ in $\DAG'(b)$ that lies on a shortest path from $b$ to $v$ in $G'$. So $(a,b) \in R_b$. Thus the edge $(a,b)$ is added to $X$ by Step~\ref{rdag:step2}.
\item If $flag(b,s) = \mbox{NUM-changed}$ then, there exist at least one path from $b$ to $s$ in $G'$ that goes through $v$, and additionally there can be shortest paths from $b$ to $s$ in $G'$ that do not go through $v$. Since, $(a, b) \in \DAG'_R(s)$, the edge $(b, a)$ lies on one or both types of paths. Suppose $(b,a) \in \pi_{bs}$, which does not use $v$. Then the path $\pi_{bs}$ is a shortest
path in $G$ and hence $(a, b) \in \DAG_R(s)$. In this case, $(a,b)$ is added to $X$ by Step~\ref{rdag:step1}.
If $(a, b) \in \pi'_{bs}$ which contains $v$, then since $\pi'_{bs}$ uses the vertex $v$,
we know that $\pi'_{bs} = b \leadsto v \leadsto s$.
Hence, similar to the case (2) above, we conclude that $(a,b) \in R_b$.
Therefore the edge $(a, b)$ is added to $X$ by Step~\ref{rdag:step2}.
\end{enumerate}
Suppose  edge $(a,b) \in X$. To show that $(a, b) \in \DAG'_R(s)$
we consider the different steps in Algorithm~\ref{algo:update-rdag} where $(a,b)$ can be added
to $X$:
\begin{enumerate}
\item The edge $(a,b)$ is added to $X$ by Step~\ref{rdag:step1}. This implies that the edge ${(a,b) \in \DAG_R(s)}$. Thus, there exists a shortest path in $\Grev$ of the form $s \leadsto a \rightarrow b$. Therefore, there exists a shortest path  in $G$ from $b$ to $s$, say $\pi_{bs} = b \rightarrow a \leadsto s$. Note that we execute Step~\ref{rdag:step1} when  $flag(b,s) = \mbox{UN-changed}$ or $flag(b, s) = \mbox{NUM-changed}$.
For either value of the flag every shortest path from $b$ to $s$ in $G$ is also shortest path in $G'$.
Therefore, the path ${\pi'_{sb} = s \leadsto a \rightarrow b}$ is a shortest path in $\Grev'$ and hence the edge $(a,b) \in \DAG'_R(s)$.
\item The edge $(a,b)$ is added to $X$ by Step~\ref{rdag:step2}. Thus, $(a,b) \in R_b$. This implies that the edge $(b,a)$ is on a shortest path in $G'$ from $b$ to $v$. Moreover, we add $(a,b)$ in Step~\ref{rdag:step2} when $flag(b, s) = \mbox{NUM-changed}$ or $flag(b,s) = \mbox{WT-changed}$.
Therefore, there exists at least one shortest path from $b$ to $s$ in $G'$ that goes through $v$. Since 
the edge $(b,a)$ is on a shortest path from $b$ to $v$ in $G'$, the edge $(b, a)$  lies on at least one shortest path from $b$ to $s$ in $G'$. 
Therefore, $(a,b)$ lies on at least one shortest path from $s$ to $b$ in $\Grev'$, and this establishes that $(a,b) \in \DAG'_R(s)$.
\end{enumerate}
This completes the proof of the lemma.
\end{appendix}

\end{document}